\documentclass[conference,letterpaper]{IEEEtran}
\usepackage[utf8]{inputenc} 
\usepackage[T1]{fontenc}
\usepackage{url}
\usepackage{ifthen}
\usepackage{cite}
\usepackage[cmex10]{amsmath}
\usepackage{comment}
\usepackage{graphicx}
\usepackage{amsmath}
\usepackage{array}
\usepackage{amsmath}%
\usepackage{cite}
\usepackage{amsfonts}%
\usepackage{amssymb}%
\usepackage{graphicx}
\usepackage{authblk}
\usepackage[margin=2cm]{geometry}
\usepackage{subfig}%
\usepackage{color}
\usepackage{dsfont}
\usepackage{comment} 
\usepackage{algpseudocode}

\newtheorem{lemma}{Lemma}

\newenvironment{proof}[1][Proof]{\textbf{#1.} }{\ \rule{0.5em}{0.5em}}
\IEEEoverridecommandlockouts
\allowdisplaybreaks

\begin{document}
\title{Coordinated Anti-Jamming Resilience in Swarm
Networks via Multi-Agent Reinforcement Learning} 

\author{Bahman Abolhassani}
\author{Tugba Erpek}
\author{Kemal Davaslioglu}
\author{Yalin E. Sagduyu}
\author{Sastry Kompella}
\affil{\normalsize Nexcepta, Gaithersburg, MD, USA 
	\thanks{ This material is based upon work supported by the ASA(ALT) SBIR CCoE under Contract No. W51701-24-C-0151.}
\vspace{-.1in}}


\maketitle

\begin{abstract}
Reactive jammers pose a severe security threat to robotic-swarm networks by selectively disrupting inter-agent communications and undermining formation integrity and mission success. Conventional countermeasures such as fixed power control or static channel hopping are largely ineffective against such adaptive adversaries. This paper presents a multi-agent reinforcement learning (MARL) framework based on the QMIX algorithm to improve the resilience of swarm communications under reactive jamming. We consider a network of multiple transmitter–receiver pairs sharing channels while a reactive jammer with Markovian threshold dynamics senses aggregate power and reacts accordingly. Each agent jointly selects transmit frequency (channel) and power, and QMIX learns a centralized but factorizable action-value function that enables coordinated yet decentralized execution. We benchmark QMIX against a genie-aided optimal policy in a no–channel-reuse setting, and against local Upper Confidence Bound (UCB) and a stateless reactive policy in a more general fading regime with channel reuse enabled. Simulation results show that QMIX rapidly converges to cooperative policies that nearly match the genie-aided bound, while achieving higher throughput and lower jamming incidence than the baselines, thereby demonstrating MARL’s effectiveness for securing autonomous swarms in contested environments.
\end{abstract}

\vspace{-.15in}
\section{Introduction}
\label{sec:Intro}

Wireless networks operating in contested environments face an escalating threat from reactive jammers that sense ongoing transmissions and inject interference only when it is most disruptive. Reactive jammers can also adapt their jamming strategies to spectrum dynamics. This adaptive approach increases jamming effectiveness while reducing the jammer’s own energy use and detectability. As spectrum becomes more crowded and mission-critical communications depend on agile links, there is a pressing need for anti-jamming countermeasures that are themselves adaptive and cooperative.

In many emerging applications such as cooperative unmanned aerial vehicles (UAVs) for search-and-rescue, autonomous ground robot teams for perimeter security, distributed sensor networks for environmental monitoring, and coordinated unmanned maritime systems, communication is not necessarily carried out by a single transmitter, but by swarms of collaborating agents. In such settings, security in swarm-robotic networks becomes a fundamental concern: the communication links between agents are vital not only for coordinating tasks (formation control, sensing, cooperative exploration) but also for maintaining resilience against adversarial interference. Because jammers can disrupt formation coherence, sever information flow, or even partition the swarm, anti-jamming strategies strengthened by multi-agent coordination are essential for resilient communications.

Traditional anti-jamming methods such as static channel hopping, fixed power control, or threshold-based policies are largely rule-driven and optimized for stationary adversaries. While these approaches can handle simple interference, they fail to anticipate and outmaneuver an intelligent jammer whose future actions depend on the defender’s behavior. Reinforcement learning (RL) has therefore emerged as a promising way to frame anti-jamming as a sequential decision problem, enabling transmitters to learn optimal signaling strategies online without prior knowledge of the jammer’s tactics.

RL has been extensively applied to anti-jamming problems in wireless communications \cite{abuzainab2019qos,li2023deepqi,qi2024deep,cdl_mitigation2024,xiao2020rl_edge,wang2021uav_jam,han2017twodim} including deep reinforcement learning (DRL) approaches such as DQN, actor–critic, and continual learning to handle richer state spaces and nonstationary adversaries. Some RL-based methods have specifically addressed reactive jammers, modeling their behavior as part of the environment rather than as a stationary interferer \cite{milcom2025_rl_reactive,zhou2020jamsa,li2022hidden_reactive,zhang2023deceiving,altun2025ofdmim,liu2024hierarchical,pourranjbar2021deceive}. However, most existing RL formulations treat each agent
independently and ignore the potential benefits of coordinated
decision-making across multiple transmitters.

In this paper we go beyond single-agent learning and propose a multi-agent reinforcement learning (MARL) framework for combating reactive jamming in swarm networks, while supporting channel access of multiple transmitters. Each transmitter must simultaneously choose a frequency band (channel) and a discrete power level for each time slot. The jammer continuously senses the total power on its current channel and, based on whether the sensed power exceeds or falls below a threshold, probabilistically decides to stay on the same channel with a high or low detection threshold or to hop to another channel with a new threshold. This Markovian behavior makes the jammer’s state dependent on the recent actions of all transmitters rather than fixed over time. The overall system goal is to maximize the long-term average sum throughput of all agents while respecting per-agent transmit power limits. Because the underlying environment is stochastic and reactive, this joint optimization problem cannot be solved analytically and calls for learning-based decentralized policies.

We present a MARL approach using the QMIX algorithm \cite{rashid2020qmix_jmlr} and tailor it for reactive jammer mitigation in swarms. QMIX trains decentralized policies in a centralized end-to-end fashion by employing a network that estimates joint action-values as a complex non-linear combination of per-agent values conditioned only on local observations. QMIX structurally enforces monotonicity between the joint action-value and per-agent values \cite{rashid2020qmix_jmlr}, which enables tractable maximization during off-policy learning and guarantees consistency between centralized training and decentralized execution (CTDE). In deployment, each agent acts based solely on local observations but in a way that optimizes the joint objective of the team. Thus, the agents can adapt their channel or power allocations in a coordinated fashion while coping with a jammer whose behavior depends on the swarm’s actions. 

By explicitly modeling how transmit actions influence future jamming states, the proposed approach treats anti-jamming as a true Markov Decision Process (MDP) and enables agents to discover cooperative policies that cannot emerge from independent or stateless learning. In addition to jammer mitigation, QMIX coordinates channel access among multiple transmitter–receiver pairs, allowing the swarm network to share spectrum efficiently while resisting interference. Unlike prior MARL studies on static jammers \cite{lv2023uavswarm_jamming}, we model reactive jammers with Markovian dynamics and employ QMIX for coordinated channel–power control under CTDE.

To assess performance, we adopt a two-stage evaluation aligned with the system assumptions. First, we study a baseline regime with constant path-loss channels (no fading) and no channel reuse. In this setting, we derive a genie-aided oracle that perfectly coordinates agents and avoids the reactive jammer, yielding an upper bound on per-agent throughput. We then compare the learned QMIX policy against this oracle across different numbers of channels, showing that QMIX closely approaches the genie benchmark while using only local observations at execution time.

Second, we generalize to a more realistic regime that allows channel reuse and incorporates small-scale Rayleigh block-fading over a coherence interval. To promote spatial reuse among distant agents, we augment the reward with a distance-aware co-channel penalty. In this regime we benchmark QMIX against two lightweight, non-learning policies that rely solely on local sensing: (i) a local Upper Confidence Bound (UCB) strategy (UCB-based channel selection with power adaptation to sensed interference) and (ii) a stateless reactive heuristic (stay-or-switch based on immediate sensing and the last reward). Across diverse topologies and increasing contention, QMIX consistently outperforms both baselines, exhibits lower co-channel interference and jamming incidence, and maintains strong decentralized performance when training is paused, demonstrating robustness to scale and harsher conditions.

The remainder of the paper is organized as follows. Section~\ref{sec:Sys_Model}  describes the system model and reactive jammer dynamics. Section~\ref{sec:learning-based} presents the MARL formulation and the QMIX training procedure under CTDE. Section~\ref{sec:baseline_model} develops the no-reuse baseline and compares QMIX to the genie-aided upper bound. Section~\ref{sec:general_model} introduces channel reuse under Rayleigh fading, defines the rule-based baselines, and reports comparative results. Section~\ref{sec:conclusion} concludes the paper.

\section{System Model}
\label{sec:Sys_Model} 
We consider a wireless network with $N$ transmitter--receiver pairs (agents) and a single reactive jammer, as shown in Fig.~\ref{fig:system-model}. Each transmitter communicates exclusively with its dedicated receiver, and hence the useful signal for receiver $i$ originates only from its paired transmitter $i$. Transmissions from all other transmitters, as well as the jammer, are treated as interference. 

\subsection{Signal Model} At time slot $t$, transmitter $i \in \{1,\dots,N\}$ selects a channel $c_i(t) \in \mathcal{C}$ from the set of $M$ orthogonal channels and a discrete power level $P_i(t) \in \mathcal{P}$. The received signal-to-interference-plus-noise ratio (SINR) for pair $i$ on channel $c_i(t)$ is given by 
\vspace{-.15in}
\begin{equation} 
\text{SINR}_i(t) = \frac{P_i(t) h_{TR}^{(i,i)}(t)}{\sum_{k \neq i, \, c_k(t)=c_i(t)} P_k(t) h_{TR}^{(k,i)}(t) + I_J^{(i)}(t) + \sigma^2}, 
\end{equation} 
where $h_{TR}^{(k,i)}(t)$ denotes the channel gain from transmitter $k$ to receiver $i$, and $I_J^{(i)}(t)$ is the interference caused by the jammer when it transmits on channel $c_i(t)$, given by \[ I_J^{(i)}(t) = \begin{cases} P_J h_{JR}^{(i)}(t), & \text{if jammer transmits on } c_i(t), \\ 0, & \text{otherwise}. \end{cases} \] The instantaneous throughput of agent $i$ is modeled as 

\begin{equation} 
R_i(t) = \log_2\big(1 + \text{SINR}_i(t)\big). \end{equation} 

The system throughput at time $t$ is then
\vspace{-.1in}
\begin{equation}
\vspace{-.15in}
R_{\text{sum}}(t) = \sum_{i=1}^N R_i(t). 
\end{equation} 

\begin{figure}[t]
  \centering
  \includegraphics[width=0.75\columnwidth,height=5cm]{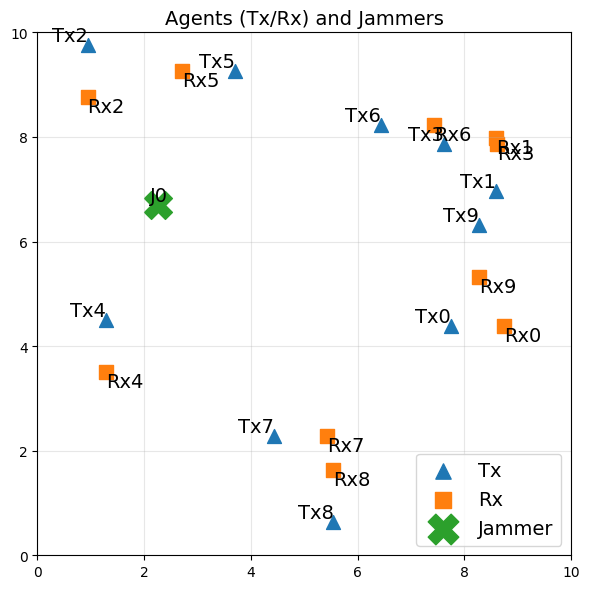}
  \caption{Network topology for $N{=}10$ agents and one jammer}
  \label{fig:system-model}
  \vspace{-.15in}
\end{figure}

\subsection{Reactive Jammer Model} 
The jammer continuously senses the aggregate \emph{received} power on its 
current channel. Let $S_J(t)$ denote the jammer’s current channel and 
$\theta(t)$ its sensing threshold. The received power at the jammer on 
channel $S_J(t)$ is 
\begin{equation}
P_J^{\mathrm{rx}}(t) = 
\sum_{i:\, c_i(t)=S_J(t)} P_i(t)\, h_{TJ}^{(i)}(t),
\end{equation}
where $h_{TJ}^{(i)}(t)$ is the channel gain from transmitter $i$ to the 
jammer. In the no-fading baseline (Section~\ref{sec:baseline_model}), 
this gain is normalized to~1.

The jammer follows a Markovian threshold-based reactive policy:
\begin{itemize}

\item If $P_J^{\mathrm{rx}}(t) \ge \theta(t)$ (threshold exceeded), the 
jammer is \emph{triggered}:
\begin{itemize}
\item With probability $p$, it stays on the same channel and sets 
$\theta(t{+}1)=\theta_H$.
\item With probability $1-p$, it hops to the next channel in cyclic 
order and sets $\theta(t{+}1)=\theta_L$.
\end{itemize}

\item If $P_J^{\mathrm{rx}}(t) < \theta(t)$ (threshold not exceeded):
\begin{itemize}
\item With probability $q$, it stays on the same channel and sets 
$\theta(t{+}1)=\theta_L$.
\item With probability $1-q$, it hops to the next channel and sets 
$\theta(t{+}1)=\theta_H$.
\end{itemize}

\end{itemize}

This formulation captures the jammer’s adaptive behavior, where its state $(S_J(t), \theta(t))$ evolves as a finite-state Markov chain influenced by the transmitters’ joint actions.

\vspace{-.05in}
\subsection{Optimization Objective} 
The main design goal is to maximize the long-term average network throughput by learning adaptive channel and power allocation strategies for the $N$ transmitter–receiver pairs:

\vspace{-.15in}
\begin{align} 
\max_{\{c_i(t), P_i(t)\}} \quad & \lim_{T \to \infty} \frac{1}{T} \sum_{t=1}^T R_{\text{sum}}(t) \\ \text{s.t.} \quad & c_i(t) \in \mathcal{C}, \quad \forall i,t, \\ & P_i(t) \in \mathcal{P}, \quad \forall i,t, \\ & \mathbb{E}[P_i(t)] \leq P_{\max}, \quad \forall i. 
\end{align} 

The constraints ensure that each transmitter selects a valid channel and power level at every slot, while respecting average transmit power limits. Due to the stochastic and reactive nature of the jammer, this optimization problem is non-convex and cannot be solved analytically. Instead, we employ MARL methods to approximate optimal decentralized policies.

\vspace{-.01in}
\section{MARL Approach} 
\label{sec:learning-based}

We formulate the adaptive channel and power allocation problem as a MARL task under the CTDE paradigm. Each transmitter--receiver pair is modeled as an autonomous agent with a deep Q-network (DQN) \emph{that observes only local information}, while a centralized mixer network (QMIX) coordinates training across agents; during execution, each agent acts solely on its local observations.

\subsection{Local Agent Observations and Actions}
At each time slot $t$, agent $i$ has access only to local information, namely the total sensed power on each channel at its own receiver, the action it selected in the previous slot (channel and power), and the immediate reward corresponding to the achieved throughput. Based on this local observation,  agent $i$ selects a joint action $a_i(t) = (c_i(t), P_i(t))$, which determines its transmission channel and power level. This decentralized execution framework reflects the constraints of practical wireless systems, where global system information is not directly available to individual nodes.

\paragraph*{MDP tuples under CTDE}
Under CTDE, each agent $i$ observes local tuples $(o_i(t), a_i(t), r_i(t), o_i(t{+}1))$ where
$o_i(t)$ stacks the per-channel received power at receiver $i$, the previous action $a_i(t{-}1)$, and the previous reward $r_i(t{-}1)$, and
$a_i(t) = (c_i(t),P_i(t))$.
The per-agent reward $r_i(t)$ is defined in terms of the instantaneous throughput $R_i(t)$ minus a co-channel interference penalty.
During centralized training, we additionally construct joint tuples $(s_t,\mathbf{a}_t,r_t,s_{t+1})$, where $s_t$ contains global features such as aggregate per-channel transmit power and the jammer state, $\mathbf{a}_t = (a_1(t),\dots,a_N(t))$, and $r_t = \sum_{i=1}^N r_i(t)$ is the team reward.
The centralized mixer network only uses $s_t$ during training; at execution time, it is discarded and all decisions are made in a decentralized POMDP setting based on $\{o_i(t)\}$.

\subsection{Centralized Training with QMIX}
During training, a centralized \emph{mixer network} aggregates the individual Q-values of all agents into a global joint action-value function $Q_{\text{tot}}$. Unlike the local agents, the mixer has access to global state features, such as aggregate per-channel transmit power and jammer state, that are unavailable during execution. This additional information enables more effective coordination among agents during training.

The QMIX architecture enforces a monotonicity constraint, expressed as $\frac{\partial Q_{\text{tot}}}{\partial Q_i} \geq 0$, which ensures that $Q_{\text{tot}}$ can be represented as a monotonic combination of the individual agent Q-values $\{Q_i\}$. As a result, an increase in any agent’s Q-value cannot reduce the global Q-function, thereby guaranteeing consistent credit assignment across agents while still allowing decentralized execution at test time.

Training proceeds by backpropagating the temporal-difference loss of the global Q-function through the mixer into each agent’s local Q-network. Although execution relies only on local observations, policies are shaped cooperatively during training to align with the global objective. Consequently, agents learn to coordinate their channel selections to minimize mutual interference, avoid persistent collisions with nearby transmitters, and adapt transmission strategies in response to the reactive jammer’s dynamics, ultimately maximizing long-term network throughput.

\subsection{Execution Phase}
Once training is complete, each agent executes its learned policy independently using only local observations, and the centralized mixer is no longer required. This ensures scalability and distributed operation in practical network deployments.

To track environmental drift, we optionally continue off-policy learning during deployment. Agents log compact summaries of local experience and periodically upload them to a centralized replay buffer. The server aggregates these logs, performs batched QMIX updates, and pushes refreshed network parameters back to agents at configurable intervals. Action selection remains fully decentralized and low-latency; in harder regimes (e.g., smaller $M$), more frequent synchronizations help sustain coordination and jammer avoidance.

To evaluate performance, we first consider a simplified model where channel gains are constant and determined solely by distance-dependent path loss, without random fading. In this case, channel reuse is disallowed by penalizing simultaneous access in the reward function, yielding an analytically tractable benchmark cost that serves as a baseline for comparison with the proposed QMIX policy.

We then extend the model to allow channel reuse when agents are sufficiently far apart so that interference remains limited, thereby enabling higher cumulative rewards. In this general setting, we introduce two rule-based policies that serve as additional baselines. These benchmarks allow us to assess the robustness of QMIX under both slow and fast Rayleigh fading environments.

\section{Baseline Model without Channel Reuse}
\label{sec:baseline_model}
We first evaluate a simplified scenario in which channel gains are constant and determined solely by distance-dependent path loss (no random fading), and channel reuse among agents is disallowed. Under this setting, we can characterize an oracle benchmark that achieves the optimal per-agent throughput via perfect coordination and jammer avoidance. This benchmark provides an upper bound for decentralized policies and serves as a reference for the proposed QMIX approach.

\paragraph*{Genie-aided oracle under no reuse} To make the baseline precise, we formalize the best possible (oracle) strategy when channels have constant path loss, reuse is not allowed (at most one agent per channel), and there is a single reactive jammer occupying one channel per slot.

\begin{lemma}[Oracle allocation with one reactive jammer] \label{lem:oracle_noreuse} 
Consider $M\!\le\!N$ channels and $N$ agents with per-slot power cap $P_{\max}$, noise variance $\sigma^2$, and no channel reuse. The jammer senses aggregate power on its channel and triggers when the total exceeds a threshold $\theta\in\{\theta_L,\theta_H\}$ with $\Pr\{\theta=\theta_L\}=q$ and $\Pr\{\theta=\theta_H\}=1-q$. Under symmetric direct gains (normalized to~1), the throughput-maximizing oracle: (i) assigns $M\!-\!1$ agents to the non-jammed channels at power $P_{\max}$, (ii) assigns one agent to the jammed channel using a power strictly below the instantaneous threshold (denoted $P_\theta<\theta$), and (iii) idles the remaining $N\!-\!M$ agents (if any). The resulting expected average per-agent rate is 
\vspace{-.01in}
\begin{equation} 
\vspace{-.01in}
\label{eq:Rstar_oracle} 
\begin{split} 
R^{\star} &= \frac{M-1}{N}\,\log_2\!\Bigl(1+\tfrac{P_{\max}}{\sigma^2}\Bigr) \\ &\quad + \frac{1}{N}\!\left[ q\,\log_2\!\Bigl(1+\tfrac{P_{L}}{\sigma^2}\Bigr) + (1-q)\,\log_2\!\Bigl(1+\tfrac{P_{H}}{\sigma^2}\Bigr) \right]. 
\end{split} 
\end{equation}
where $P_L<\theta_L$ and $P_H<\theta_H$ are the conservative powers used on the jammed channel when the jammer employs thresholds $\theta_L$ and $\theta_H$, respectively. 
\end{lemma}

\begin{proof}[Sketch] With no reuse, inter-agent interference vanishes when each active channel is occupied by at most one agent. Since $\log_2(1+\mathrm{SNR})$ is increasing in transmit power, every safe (non-jammed) channel should be filled at $P_{\max}$, giving the first term in \eqref{eq:Rstar_oracle}. On the jammed channel, transmitting above the current threshold triggers the jammer and collapses the rate; hence the oracle always uses power strictly below the threshold. Averaging over the two threshold states yields the second term. Reassigning any of the $M{-}1$ safe-channel agents to idle, or moving them to the jammed channel, weakly decreases the sum rate; activating more than $M$ agents violates the no-reuse constraint. Thus, the stated allocation and rate are optimal under the model assumptions. 
\end{proof}

\noindent\textit{Remark.} 
Equation~\eqref{eq:Rstar_oracle} can be instantiated with the normalized powers used in our experiments (e.g., $P_L$ and $P_H$ corresponding to conservative settings under $\theta_L$ and $\theta_H$), yielding the closed-form oracle curves shown alongside our QMIX results. Intuitively, to achieve $R^\star$ the controller must (i) place one user per safe channel at $P_{\max}$, and (ii) on the jammed channel, transmit just below the threshold; the remaining agents must idle when $M<N$.

We consider $N=10$ transmitter--receiver pairs (shown 
in Fig.~\ref{fig:topology10}) and vary the number of channels $M \in \{10,8,4\}$. The jammer employs two sensing thresholds, a high level $\theta_H\!=\!0.4$ and a low level $\theta_L\!=\!0.2$. In Fig.~\ref{fig:baseline_subplots}, panels (a)–(c), the orange curve shows the average reward per agent per slot during training. The blue curve represents the penalized reward, where each pair of interfering agents reduces the reward by one. The convergence of the blue curve to the orange curve demonstrates that the agents successfully learn to avoid mutual interference. In addition, the green markers indicate the average throughput per agent once training is paused and agents rely solely on their local information to decide channel and power levels. These points highlight that the agents acquire a robust policy that fully avoids interference and adapts to the jammer using only local observations.

\begin{figure*}[t]
\centering
\subfloat[$M=10$ channels\label{fig:based10}]{
    \includegraphics[width=0.3\textwidth]{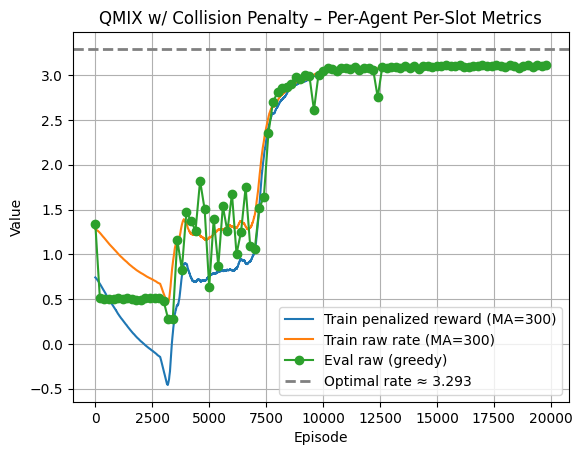}
}
\hfill
\subfloat[$M=8$ channels\label{fig:based8}]{
    \includegraphics[width=0.3\textwidth]{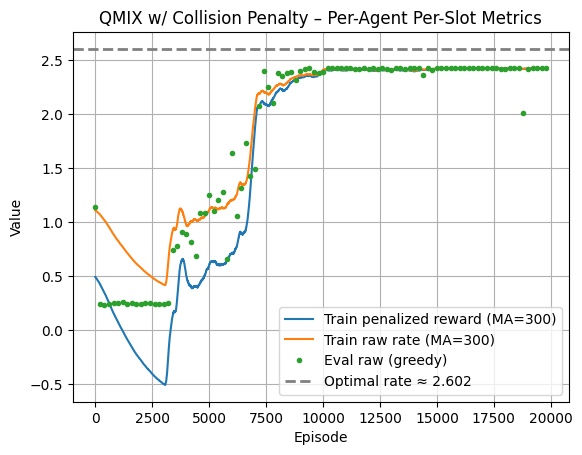}
}
\hfill
\subfloat[$M=4$ channels\label{fig:based4}]{
    \includegraphics[width=0.3\textwidth]{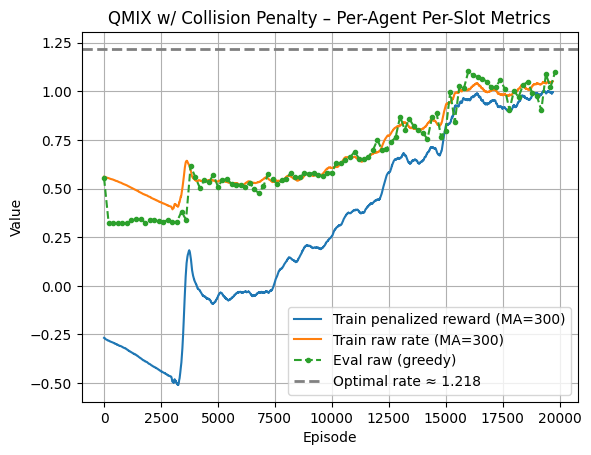}
}
\caption{Comparison for $N=10$ agents: QMIX vs. oracle throughput under perfect coordination and jammer avoidance. Orange: average reward during training (raw rate); Blue: penalized reward accounting for interference; Green: average throughput per agent under decentralized execution with only local information.}
\label{fig:baseline_subplots}
\vspace{-.1in}
\end{figure*}

As the number of channels increases, contention diminishes and QMIX closely approaches the oracle benchmark, indicating effective coordination and jammer avoidance. Conversely, when $M$ is smaller (e.g., $M=4$), the gap to the oracle widens and the learning curves exhibit longer transients with a lower asymptote. This slower convergence reflects the greater coordination burden among agents and the need to infer the jammer’s behavior under tighter spectrum reuse. Notably, the learned policies are \emph{not} static channel assignments: agents follow stochastic, observation-driven strategies that opportunistically switch channels based on local sensing, leading to slot-to-slot variation rather than fixed per-agent channels. Nonetheless, QMIX continues to sustain a significant fraction of the oracle performance, confirming its ability to learn robust decentralized policies even under stringent channel constraints.

\vspace{-.00in}
\section{Generalized Model with Channel Reuse under Rayleigh Fading}
\label{sec:general_model}

We next extend the model to allow channel reuse among agents that are sufficiently far apart, so that interference remains limited and overall throughput can be improved. In addition, the channel model is generalized to incorporate small-scale Rayleigh fading in addition to distance-dependent path loss we considered so far. The channel gain from transmitter $i$ to receiver $j$ is 
\vspace{-.1in}
\begin{equation} 
\vspace{-.1in}
h_{TR}^{(i,j)}(t) = \frac{\tilde{h}_{TR}^{(i,j)}(t)}{d_{ij}^{\alpha}}, 
\end{equation} 
where $d_{ij}$ is the distance between nodes, $\alpha$ is the path loss exponent, and $\tilde{h}_{TR}^{(i,j)}(t)$ is a Rayleigh fading coefficient with unit mean power. Similarly, the jammer-to-receiver and transmitter-to-jammer links are modeled as 
\begin{align} 
h_{JR}^{(i)}(t) 
= \frac{\tilde{h}_{JR}^{(i)}(t)}{d_{J,i}^{\alpha}}, \quad \text{and}\quad h_{TJ}^{(i)}(t) 
= \frac{\tilde{h}_{TJ}^{(i)}(t)}{d_{i,J}^{\alpha}}. \end{align}
All fading terms $\tilde{h}(t)$ are independent circularly symmetric complex Gaussian random variables with zero mean and unit variance, corresponding to Rayleigh fading. A block-fading model is assumed, where coefficients remain constant over a coherence interval of $T_c$ slots and change independently across blocks.

To benchmark the proposed QMIX approach, we also implement two lightweight rule-based policies that rely only on local sensing information.

\textbf{UCB with Power Adaptation:} Each agent maintains statistics for every channel $j$ (number of plays $n_j$, average reward $\mu_j$) and computes 
\vspace{-.15in} 
\begin{equation} 
I_j = \mu_j + c \sqrt{\tfrac{\ln t}{n_j}}, 
\end{equation} 
where $c>0$ controls exploration. The selected channel is $c^\star=\operatorname*{arg\,max}_{j\in\mathcal{C}} I_j$, and transmit power is adapted to sensed interference $s$ via $P(s)=\{\,1.0 \text{ if } s\le 0.15;\; 0.7 \text{ if } 0.15<s\le 0.35;\; 0.4 \text{ if } s>0.35\,\}.$

\textbf{Stateless Heuristic:} The agent remains on its previous channel if the sensed interference is below a threshold ($0.2$) and the previous reward exceeds $0.1$; otherwise, it switches to the channel with the lowest interference. Power is again chosen according to the same interference thresholds as above.

These two baselines represent different trade-offs: the UCB policy balances exploration and exploitation, while the stateless heuristic is purely reactive. They serve as useful comparison points to assess the robustness and performance gains of QMIX under channel reuse and fading.

\paragraph*{Distance-aware interference penalty} We refine the reward to penalize co-channel interference proportionally to distance, so farther pairs incur smaller penalties. Specifically, if agents $i$ and $j$ use the same channel at slot $t$, the interference penalty added to the global cost is
\begin{equation}
\vspace{-.1in}
\label{eq:dist_penalty}
\Delta r_{ij}(t) \;=\; \frac{\lambda}{d_{ij}^{\beta}},
\end{equation}
where $\lambda>0$ and $\beta\ge1$ are tunable parameters. This encourages channel reuse among distant agents while discouraging harmful co-channel clustering.

We use Rayleigh block-fading with coherence $T_c=100$ slots (one episode); channel coefficients remain constant within an episode and are i.i.d.\ across episodes. The distance-aware penalty in \eqref{eq:dist_penalty} is added to the reward during training to encourage faster convergence and promote throughput-enhancing channel reuse.

Fig.~\ref{fig:topos} illustrates the layout and Figs.~\ref{fig:perf5}-\ref{fig:perf10} present learning curves. Orange line shows the average reward per agent per slot during training; blue line depicts the penalized reward using the distance-aware interference term; and green markers are the average per-agent throughput under decentralized execution when training is paused (agents use only local observations).

\begin{figure}[t]
\centering
\subfloat[$N{=}5$ topology\label{fig:topology5}]{
  \includegraphics[width=0.47\columnwidth]{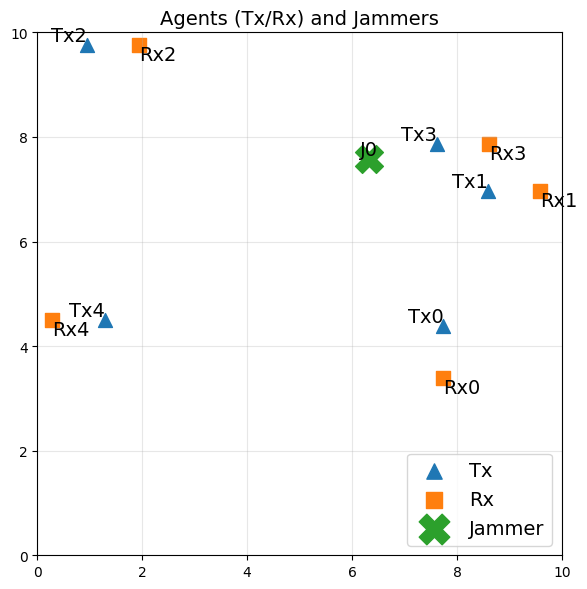}
}\hfill
\subfloat[$N{=}10$ topology\label{fig:topology10}]{
  \includegraphics[width=0.47\columnwidth]{topology.png}
}
\caption{Spatial layouts for $M{=}4$ channels and one jammer.}
\label{fig:topos}
\vspace{-.15in}
\end{figure}

\begin{figure}[t]
\centering
\includegraphics[width=0.8\columnwidth]{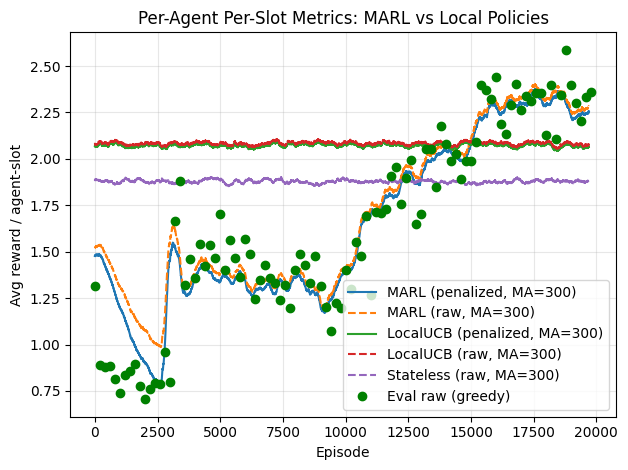}
\caption{Average reward per agent per slot for $N{=}5$, $M{=}4$, one jammer.}
\label{fig:perf5}
\vspace{-.15in}
\end{figure}

\begin{figure}[t]
\centering
\includegraphics[width=0.8\columnwidth]{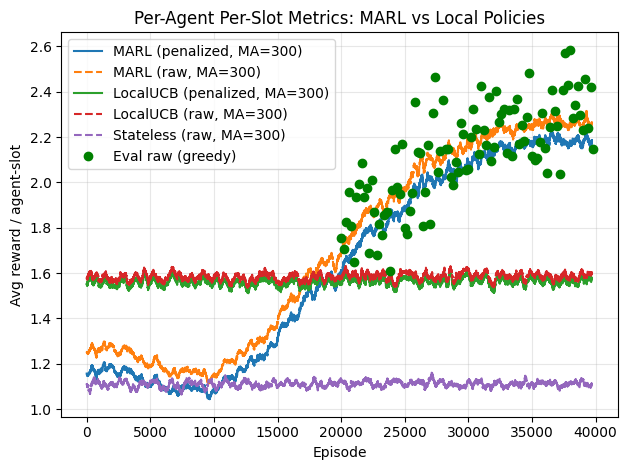}
\caption{Average reward per agent per slot for $N{=}10$, $M{=}4$, one jammer.}
\label{fig:perf10}
\vspace{-.15in}
\end{figure}

In Fig.~\ref{fig:perf5}, QMIX requires an initial learning phase but converges to the best performance, surpassing both local UCB and the stateless heuristic; the blue curve tracking the orange curve indicates that agents learn to suppress harmful co-channel interference under the distance-aware penalty.

In Fig.~\ref{fig:perf10}, the spectrum remains fixed at $M=4$ channels while the number of agents increases from $N=5$ to $N=10$, substantially raising contention and jammer exposure. Despite the harsher setting, QMIX scales robustly and further widens its margin over local UCB and the stateless heuristic (exceeding 50\% at convergence). The penalized reward (blue) continues to closely track the training reward (orange), indicating sustained suppression of co-channel interference under the distance-aware penalty, while the execution markers (green) confirm that these gains persist under fully decentralized operation using only local observations.

\vspace{-.00in}
\section{Conclusion} 
\label{sec:conclusion} 
This paper studied the problem of coordinated anti-jam decision making in swarm networks, where multiple transmitter–receiver pairs must sustain reliable connectivity in the presence of a reactive jammer that may adapt its jamming decisions. We proposed a MARL framework based on QMIX that enables decentralized agents to cooperate through a centralized training process while executing locally. Unlike single-agent or stateless policies, QMIX explicitly accounts for the impact of joint transmit actions on future jamming states, thereby treating anti-jamming as a true MDP. We demonstrated that QMIX converges rapidly to cooperative strategies that approach the performance of a genie-aided benchmark in which all agents have prior knowledge of jamming actions. Moreover, QMIX consistently outperformed rule-based baselines, including LocalUCB and a Stateless reactive policy, by achieving higher throughput and reducing the success rate of jamming attacks across a variety of network and jammer conditions. In addition, we showed that QMIX can effectively coordinate channel access among multiple transmitter–receiver pairs, balancing spectrum efficiency with resilience. These results highlight the promise of MARL as a practical and scalable tool for ensuring swarm resilience against adaptive jammers. Future work will extend this framework to larger-scale swarms with heterogeneous agents, investigate alternative MARL architectures and learning-based or multi-jammer adversaries, and incorporate additional dimensions such as latency, energy constraints, and noisy sensing.

\bibliographystyle{IEEEtran}
\bibliography{new_references}
\end{document}